\documentclass{article}

\usepackage{arxiv}

\usepackage[utf8]{inputenc} 
\usepackage[T1]{fontenc}    
\usepackage{hyperref}       
\usepackage{url}            
\usepackage{booktabs}       
\usepackage{amsfonts}       
\usepackage{nicefrac}       
\usepackage{microtype}      
\usepackage{lipsum}
\usepackage{graphicx}
\graphicspath{ {./images/} }

\usepackage{amsmath}
\usepackage{amssymb}
\usepackage{amsfonts}
\usepackage{amstext}
\usepackage{amsthm}
\usepackage{setspace}
\usepackage{algorithm}
\usepackage{algorithmic}

\usepackage{xcolor}
\usepackage{cite}
\newtheorem{theorem}{Theorem}
\newtheorem{corollary}{Corollary}[theorem]

\usepackage[english]{babel}
\theoremstyle{definition}
\newtheorem{definition}{Definition}
\theoremstyle{example}
\newtheorem{example}{Example}

\title{On the Computation Rate of All-Reduce}

\author{Yufeng Zhou \\
        Department of Electrical Engineering \\
        University of North Texas, Denton TX \\
        \texttt{yufengzhou@my.unt.edu} \\
        \And
        Hua Sun \\
        Department of Electrical Engineering \\
        University of North Texas, Denton TX \\
        \texttt{hua.sun@unt.edu} \\
}

\begin{document}

\maketitle

\begin{abstract}
In the All-Reduce problem, each one of the $K$ nodes holds an input and wishes to compute the sum of all $K$ inputs through a communication network where each pair of nodes is connected by a parallel link with arbitrary bandwidth. The computation rate of All-Reduce is defined as the number of sum instances that can be computed over each network use. For the computation rate, we provide a cut-set upper bound and a linear programming lower bound based on time (bandwidth) sharing over all schemes that first perform Reduce (aggregating all inputs at one node) and then perform Broadcast (sending the sum from that node to all other nodes). Specializing the two general bounds gives us the optimal computation rate for a class of communication networks and the best-known rate bounds (where the upper bound is no more than twice of the lower bound) for cyclic, complete, and hypercube networks. 
\end{abstract}

\section{Introduction}
The training of modern massive scale learning models places a prohibitively high demand on the underlying multi-node communication networks to efficiently and reliably compute and propagate the desired functions and simultaneously achieve high speed and low latency. 
The involved communication for computation tasks are implemented by calling some basic building blocks referred to as collective communication primitives \cite{rabenseifner2004optimization, thakur2005optimization, sanders2019sequential, wang2020blink}, including All-Gather, All-Reduce, Broadcast, Reduce, and Reduce-Scatter. Among these primitives, All-Reduce is arguably the most important as it is frequently encountered and most complex, occupies the largest fraction of communication resource, and is often the bottleneck of the overall system performance \cite{patarasuk2009bandwidth, patarasuk2007bandwidth, kolmakov2020generalization, luczynski2024near, de2024swing}. The goal of this work is to study the ultimate communication efficiency of All-Reduce through a communication theory (computation rate) perspective.

In modern data centers, processing units are connected through wired parallel links. In this work, we consider a communication network where each pair of nodes is connected by a noiseless link with arbitrary bandwidth. This is a very general model, as any (one-hop) network topology can be obtained by setting certain links to have zero bandwidth. For the All-Reduce primitive, each node has an input and wishes to use the underlying communication network to compute the sum of the inputs from all nodes. The performance metric considered in this work is the computation rate, i.e., the ratio of the number of sum instances that can be computed over the number of times the network is used. That is, block coding is allowed where each input may contain multiple symbols and we may use the network multiple times, similar to the setting of the canonical communication rate metric considered in information theory \cite{Cover_Thomas}. All-Reduce has been mostly studied in computer science literature where different metrics are typically considered such as latency and total consumed bandwidth \cite{rabenseifner2004optimization, thakur2005optimization, sanders2019sequential, wang2020blink, patarasuk2009bandwidth, patarasuk2007bandwidth, kolmakov2020generalization, luczynski2024near, de2024swing}; while we also note that the communication primitives have started to gain attention in the communication theory community, e.g., another variant of All-Reduce is studied in \cite{li2024optimal}; All-Gather is studied in \cite{huang2025optimality}.

We next describe the main results obtained. First, we present a general upper bound and a general lower bound on the computation rate that applies to any network, i.e., any bandwidth of the links. The upper bound is based on a cut-set argument that translates the All-Reduce computation problem to a point-to-point communication problem. The cut-set argument is standard in network coding for computing literature \cite{appuswamy2011network, guang2019improved}. Simple as it seems, we are not able to improve the cut-set bound which is left as an open problem (refer to Section \ref{sec:dis} for more discussions). The lower bound is based on time sharing over all possible first-Reduce-and-then-Broadcast schemes, i.e., first computing the sum of all inputs at one node by aggregating the inputs through a sub-network that forms a spanning tree topology and then broadcasting the sum from that node to all other nodes through another spanning tree. By considering all possible spanning trees for Reduce and Broadcast, we fully exploit the potential of this class of Reduce-and-Broadcast schemes available in the communication network and the weight of each sub-scheme is optimized through a linear program to maximize the computation rate within the bandwidth constraints. While this scheme is conceptually simple, we are not able to improve the linear programming lower bound which is also left as an open problem (refer to Section \ref{sec:dis} for more discussions).

Second, we apply the cut-set upper bound and the linear programming lower bound to some specific networks with practical topology to obtain more concrete results; note that both the upper bound and the lower bound involve optimizations so that the bounds are not trivial to evaluate (especially the lower bound that needs to consider all possible spanning trees whose number is exponential). In this regard, we find a class of networks where the two bounds match, i.e., the maximum computation rate is characterized. This class of network is intimately related to the spanning tree packing scheme of Reduce-and-Broadcast and is essentially some linear combination of tree topologies (so trees are included as a special case, refer to Section \ref{sec:optimal} for details). In addition, we explicitly evaluate the upper and lower bounds for some commonly encountered network topologies (and often impose that the connected links have the same bandwidth), including cycles, complete graphs, and hypercubes. The obtained bounds are the tightest in the literature, to the best of our knowledge; for the upper bound, we are not aware of any explicit existing bounds (while the cut-set argument is well-known as noted above); for the lower bound, we have applied existing algorithms in the literature (although they are not originally proposed to maximize computation rate) and find that the rate achieved is no larger than our lower bound (some contains similar ideas on spanning tree packing \cite{wang2020blink}; refer to Section \ref{sec:ach} for details). For each network, the upper bound is no more than twice of the lower bound, so the maximum computation rate is characterized to within a multiplicative gap of $2$.

\section{System Model}
Consider $K \geq 2$ nodes where node $i$ holds an input $W_i$, $i \in \{1, 2, \cdots , K\} \triangleq [K]$. 
The inputs are independent and each input consists of $L$ i.i.d. uniform symbols from $\mathbb{F}_q$, i.e., $W_i \in \mathbb{F}_q^{L \times 1}$.
\begin{eqnarray}
&& H(W_1)=H(W_2)= \cdots =H(W_K)=L \notag \\
&& H(W_1,W_2, \cdots, W_K)=H(W_1)+H(W_2)+ \cdots +H(W_K) \label{h12}
\end{eqnarray}
where the entropy in this work is measured in $q$-ary units.

Node $i$ is connected to node $j$ with a noiseless link that can carry $\beta_{ij}$ symbols 
from $\mathbb{F}_q, i, j \in [K], i \neq j$ where $\beta_{ij}$ is assumed to be an integer with no loss of generality. We label the communication network $\mathcal{N}$ through the bandwidth of all links, i.e., $\mathcal{N}(\vec{\beta}) \triangleq \mathcal{N}(\beta_{12}; \cdots; \beta_{(K-1)K})$. Each node wishes to compute $W_1 + \cdots + W_K \in \mathbb{F}_q^{L \times 1}$ by using the network $\mathcal{N}(\vec{\beta})$ $N$ times. Over the $n$-th network use, $n \in [N]$, node $i$ sends a message $X_{ij}(n) \in \mathbb{F}_q^{\beta_{ij} \times 1}$ to node $j$ and as the links are parallel, node $j$ receives the following messages
\begin{eqnarray}
Y_j(n)=[X_{1j}(n), X_{2j}(n), \cdots, X_{ij}(n), \cdots, X_{Kj}(n)]. \label{rx}
\end{eqnarray}
$X_{ij}(n)$ is required to be a function of the input $W_i$ and the messages previously received by node $i$,
\begin{eqnarray}
H(X_{i1}(n), \cdots, 
X_{iK}(n) | W_i, Y_i(1:n-1))=0 \label{det}
\end{eqnarray}
where $Y_i(1:n-1) \triangleq (Y_i(1), Y_i(2), \dots, Y_i(n-1) )$.

At the end of the $N$-th network use, each node may recover $W_1 + \cdots + W_K$ with no error from all available information.
\begin{eqnarray}
H(W_1+\cdots+W_K | W_i, Y_i(1:N)) = 0. \label{dec}
\end{eqnarray}


The computation rate of All-Reduce characterizes the number of sum instances computed for each network use and is defined as
\begin{eqnarray}
R=\frac{L}{N}. \label{rate}
\end{eqnarray}
We are interested in characterizing the maximum computation rate, defined as $R^{*} = \sup_L R$.



\section{Results}

\subsection{Cut-Set Upper Bound}



A upper (converse) bound based on cut-set arguments of the computation rate is given in the following theorem. For $\mathcal{K} \subset [K]$, $\mathcal{K}^c$ denotes the complement set, i.e., the set of elements that belong to $[K]$ but not to $\mathcal{K}$.

\begin{theorem}\label{thm:cut}
For network $\mathcal{N}(\vec{\beta})$, 
the All-Reduce computation rate satisfies
\begin{eqnarray}
R^* \leq \overline{R} ~\triangleq~ \min_{\mathcal{K} \subsetneq [K], \mathcal{K} \neq \emptyset} \sum_{i \in \mathcal{K}, j \in \mathcal{K}^c} \beta_{ij}. \label{eq:cut}
\end{eqnarray}
\end{theorem}

Intuitively, upper bound (\ref{eq:cut}) says that the computation rate cannot exceed the information flow from nodes in $\mathcal{K}$ to their complement. It follows from a cut-set argument because if all links from nodes in $\mathcal{K}$ to their complement are cut out, All-Reduce cannot be completed. An entropy based proof is presented next.

{\it Proof:} Consider any non-empty $\mathcal{K}$ and any $s \in \mathcal{K}$. 
Allow full cooperation among nodes $\mathcal{K}$ and name them as the super source, i.e., the super source knows $(W_i)_{i \in \mathcal{K}}$ (in particular, $W_s$ included). Allow full cooperation among nodes $\mathcal{K}^c$ and name them as the super destination. Give all inputs except $W_s$ to the super destination, i.e., the super destination knows $(W_k)_{k \in [K], k \neq s}$. Allowing cooperation and making $W_k$ available cannot hurt the rate so that it suffices to upper bound the computation rate of the two-node (super source-destination) network as follows\footnote{We have now transformed the All-Reduce computation problem to a two-node communication problem where the super source wishes to send $W_s$ to the super destination and may directly apply the well-known cut-set bound \cite{Cover_Thomas}. We give a full proof here for completeness.}. Consider any feasible All-Reduce protocol. 
\begin{eqnarray}
L &\overset{(\ref{h12})}{=}& H(W_1 + \cdots + W_K) \label{e1} \\ 
&\overset{(\ref{dec})}{=}& I\left(W_1 + \cdots + W_K; \big(Y_j(1:N)\big)_{j \in \mathcal{K}^c},  (W_k)_{k \in [K], k \neq s}\right) \label{e2}  \\
&\overset{(\ref{h12})}{=}& I(W_1 + \cdots + W_K; (Y_j(1:N))_{j \in \mathcal{K}^c} | (W_k)_{k \in [K], k \neq s}) \label{e3}  \\
&\overset{(\ref{rx})}{=}& I(W_1 + \cdots + W_K; (X_{tj}(1:N))_{t \in [K], j \in \mathcal{K}^c} | (W_k)_{k \in [K], k \neq s}) \label{e4}  \\
&\overset{(\ref{det})}{=}& I(W_1 + \cdots + W_K; (X_{ij}(1:N))_{i \in \mathcal{K}, j \in \mathcal{K}^c} | (W_k)_{k \in [K], k \neq s}) \label{e5}  \\
&\leq&  H((X_{ij}(1:N))_{i \in \mathcal{K}, j \in \mathcal{K}^c}) \label{e6}  \\
&\leq& N \sum_{i \in \mathcal{K}, j \in \mathcal{K}^c} \beta_{ij} \label{e7} \\
\Rightarrow R &\overset{(\ref{rate})}{=}& L/N ~\leq~ \sum_{i \in \mathcal{K}, j \in \mathcal{K}^c} \beta_{ij} \label{e8}
\end{eqnarray}
where (\ref{e1}) follows from the fact that $W_1+\cdots+W_K$ contains $L$ uniform $\mathbb{F}_q$ symbols as the inputs are independent and uniform. (\ref{e2}) is due to the observation that the super destination may recover $W_1 + \cdots + W_K$ with no error. (\ref{e3}) follows from the independence of the inputs. In (\ref{e4}), we have plugged in the received messages and (\ref{e5}) is obtained because the messages from nodes inside the super destination is a function of other messages and inputs known to the super destination. (\ref{e7}) follows from the size of the messages and the property that uniform distribution maximizes entropy. Minimizing the right-hand-side of (\ref{e8}) over all $\mathcal{K}$ and taking supremum on both sides over $L$ gives the desired upper bound (\ref{eq:cut}).
\hfill\qed

\subsection{Linear Programming Lower Bound}

We give an achievable scheme (a lower bound) for the All-Reduce problem by time sharing over all schemes that first perform Reduce and then perform Broadcast. Reduce refers to a scheme where all nodes aggregate their inputs at one node (called the root) through a directed tree, called rooted MAC tree network defined next.


\begin{definition}[Rooted MAC Tree Network]
A network $\mathcal{N}_{T_M} (\vec{\beta}_{T_M})$ is called a rooted MAC tree network 
if the edges corresponding to the non-zero elements of the bandwidth vector $\vec{\beta}_{T_M}$ form a directed (spanning) tree topology. Specifically,
for any $i \in [K], i \neq r$, there exists $i^* \in [K]$ such that\footnote{The $K-1$ edges $i \rightarrow i^*$ form a spanning tree because there exists a directed path from any node $i \neq r$ to $r$ (by going through the edges $i \rightarrow i^*$ recursively).} $\beta_{i i^*} = 1$ and $\beta_{i j} = 0$ for all $j \in [K], j \neq i^*$. 

Node $r$ has no outgoing edges and is referred to as the root node. To make the root node $r$ explicit, we will denote a rooted MAC tree network by $\mathcal{N}_{T_M} (\vec{\beta}_{T_M}, r)$. Note that $\mathcal{N}_{T_M} (\vec{\beta}_{T_M}, r)$ has $K-1$ edges as each node except the root has one outgoing edge, i.e., $\vec{\beta}_{T_M}$ contains $K-1$ ones.
\end{definition}


Given a rooted MAC tree network $\mathcal{N}_{T_M} (\vec{\beta}_{T_M}, r)$, we may perform Reduce by using the network $K-1$ times\footnote{Often we need strictly less than $K-1$ times slots, but from a computation rate perspective there is no rate loss so no further improvements are pursued in this work.} (aggregating the inputs from the bottom of the tree to the root, one edge at a time) so that in the end, root node $r$ knows $\sum_{k=1}^K W_k$ where each $W_k$ contains $1$ symbol. With $K$ nodes, there are $K^{K-1}$ distinct rooted MAC tree networks because $K$ nodes may form $K^{K-2}$ different undirected spanning trees (a canonical result by Cayley, see Theorem 2.1 of \cite{van2001course}) and any one of the $K$ nodes may be the root so as to result in a distinct $\vec{\beta}_{T_M}$. Fig.~\ref{fig:mac} contains an example for $K=3$.





\begin{figure}[hbt!]
\centering
\includegraphics[scale=0.6]{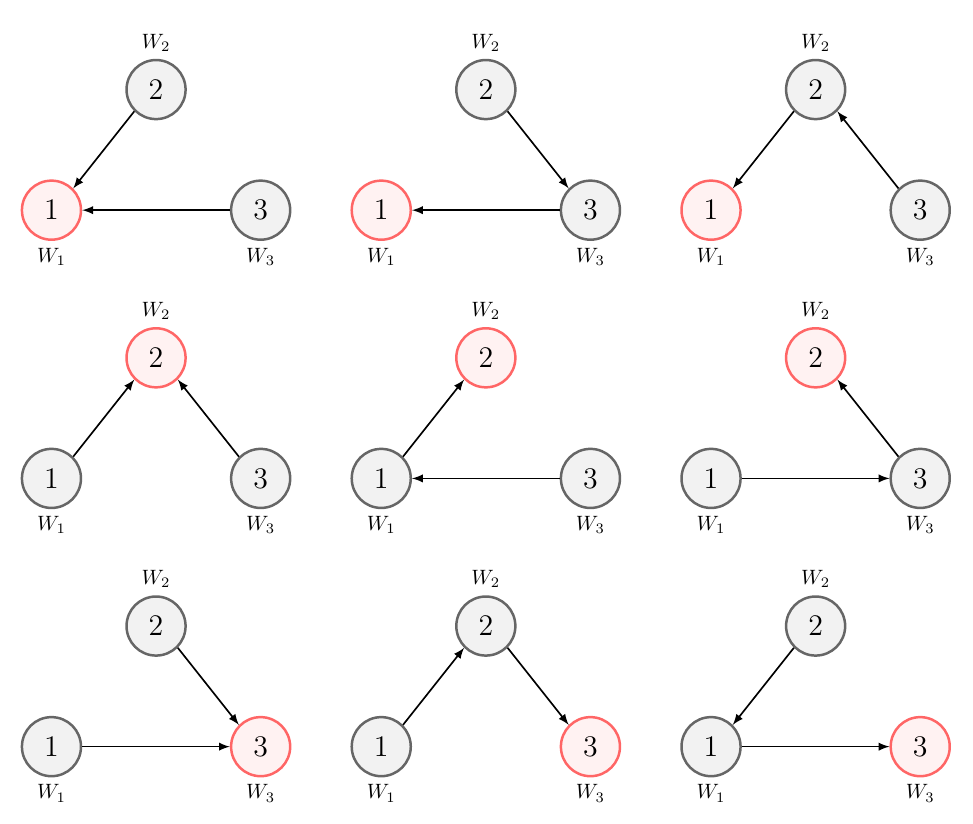}
\caption{All $3^2 = 9$ rooted MAC tree networks  $\mathcal{N}_{T_M} (\vec{\beta}_{T_M}, r)$ with $K = 3$ nodes and the root node is colored in red. For each network, the root node may compute $W_1+W_2+W_3$ by using the network $2$ times.}
\label{fig:mac}
\end{figure}

Next, we proceed to the Broadcast operation where the sum is sent from the root node to all other nodes through another directed tree, called rooted BC tree network defined next.

\begin{definition}[Rooted BC Tree Network]
A network $\mathcal{N}_{T_B} (\vec{\beta}_{T_B}, r)$ is called a rooted BC tree network if the non-zero bandwidth edges form a directed tree leaving the root node $r \in [K]$, i.e., 
for any $j \in [K], j \neq r$, there exists $j^* \in [K]$ such that $\beta_{j^* j} = 1$ and $\beta_{i j} = 0$ for all $i \in [K], i \neq j^*$. 
\end{definition}

Rooted BC tree networks are dual to rooted MAC tree networks in the sense that reversing the direction of the edges of a rooted BC tree network $\mathcal{N}_{T_B} (\vec{\beta}_{T_B}, r)$ gives a rooted MAC tree network $\mathcal{N}_{T_M} (\vec{\beta}_{T_M}, r)$, i.e., $\beta_{ij} = 1$ in $\vec{\beta}_{T_M}$ if and only if $\beta_{ji} = 1$ in $\vec{\beta}_{T_B}$. As a result of this dual property, we know that given a rooted BC tree network $\mathcal{N}_{T_B} (\vec{\beta}_{T_B}, r)$, we may perform Broadcast by using the network $K-1$ times (broadcasting the sum from the root to all other nodes, one edge at a time) so that in the end, every node knows $\sum_{k=1}^K W_k$. Similarly with $K$ nodes, there are $K^{K-1}$ distinct rooted BC tree networks (see Fig.~\ref{fig:bc} for an example with $K=3$). 

\begin{figure}[hbt!]
\centering
\includegraphics[scale=0.6]{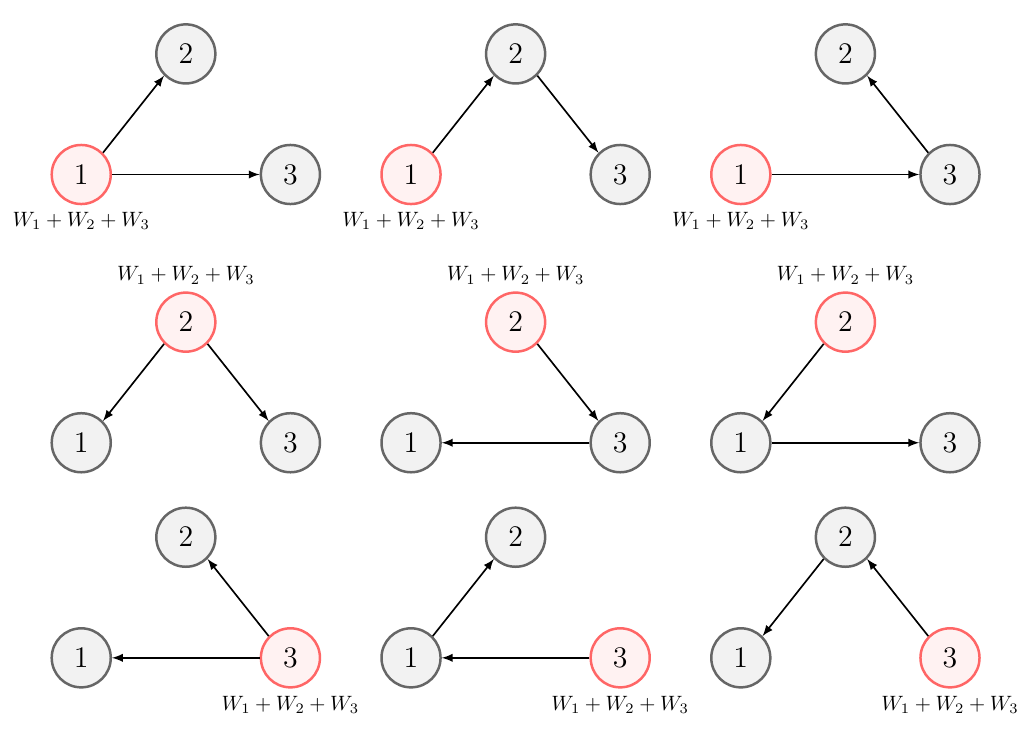}
\caption{All $3^2 = 9$ rooted BC tree networks  $\mathcal{N}_{T_B} (\vec{\beta}_{T_B}, r)$ with $K = 3$ nodes and the root node is colored in red. For each network, the root node may propagate $W_1+W_2+W_3$ to all other nodes by using the network $2$ times.}
\label{fig:bc}
\end{figure}

Equipped with rooted MAC and BC tree networks, we are ready to combine them into a rooted MAC-BC network, defined next and over which we may achieve All-Reduce computation rate $R=1$.

\begin{definition}[Rooted MAC-BC Network]
Given a rooted MAC tree network $\mathcal{N}_{T_M} (\vec{\beta}_{T_B}, r)$ and a rooted BC tree network $\mathcal{N}_{T_B} (\vec{\beta}_{T_B}, r)$ with the same root $r$ (and the same number of nodes $K$), the network $\mathcal{N}_{T_{MB}} (\vec{\beta}_{T_{MB}}, r)$ where $\vec{\beta}_{T_{MB}} = \vec{\beta}_{T_{M}} + \vec{\beta}_{T_{B}}$ is defined as a rooted MAC-BC network.
\end{definition}

Following the results of rooted MAC and BC tree networks, we know that a rooted MAC-BC network $\mathcal{N}_{T_{MB}} (\vec{\beta}_{T_{MB}}, r)$ may perform All-Reduce for $L=1$ symbol with $N = 2(K-1)$ network uses; further, we may repeat this scheme in a streamline manner to perform $L$ instances of All-Reduce with $N = 2(K-1) + L-1$ network uses for any $L$, so the achieved computation rate $R = L/N = L/(L+2K-3)$ approaches $1$ when $L$ approaches infinity.
There are $K \times (K^{K-2})^2 = K^{2K-3} \triangleq Z$ distinct rooted MAC-BC networks (see Fig.~\ref{fig:macbc} for an example with $K=3$).





\begin{figure}[h]
\centering
\includegraphics[scale=0.6]{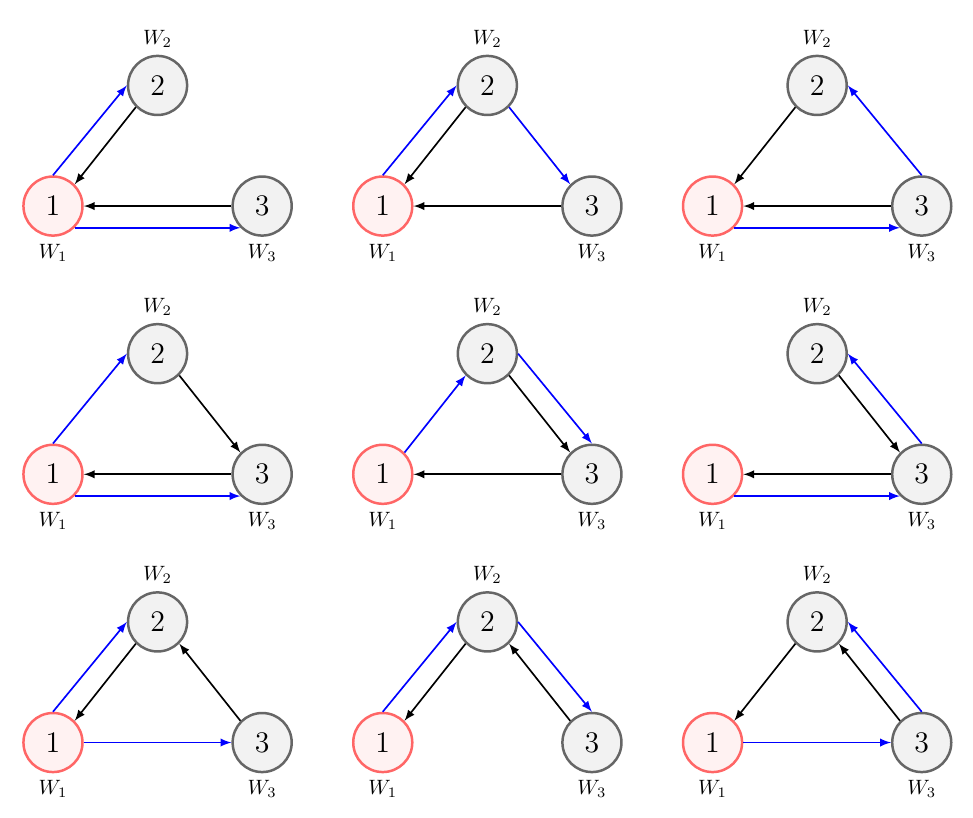}
\caption{All $9$ rooted MAC-BC tree networks $\mathcal{N}_{T_{MB}} (\vec{\beta}_{T_{MB}}, r=1)$ with $K = 3$ nodes and root $r=1$ (colored in red). Each network may perform $L$ instances of sum computation with $N=L+3$ network uses where Reduce edges are in black and Broadcast edges are in blue, achieving rate $R = L/N \rightarrow 1$ as $L \rightarrow \infty$.}
\label{fig:macbc}
\end{figure}

Finally, we jointly consider all possible rooted MAC-BC networks and maximize the computation rate through optimizing the weights of each MAC-BC network (that is able to achieve All-Reduce computation rate $1$) with a linear program that captures the bandwidth constraint for each link.

Denote the $z$-th rooted MAC-BC network by $\mathcal{N}_{T_{MB}}^z (\vec{\beta}_{T_{MB}}^z, r^z), z \in [Z]$ and suppose it is employed $\lambda_z$ times (associated with weight $\lambda_z$, i.e., occupies $\lambda_z$ times of bandwidth of that of the rooted MAC-BC network) over an arbitrary network $\mathcal{N}(\vec{\beta})$. Then to not exceed the bandwidth constraint, we have
\begin{eqnarray}
\lambda_1 \vec{\beta}_{T_{MB}}^1 + \cdots + \lambda_Z \vec{\beta}_{T_{MB}}^Z \leq \vec{\beta}
\end{eqnarray}
and the computation rate achieved is $\lambda_1 + \cdots + \lambda_Z$. Therefore the maximum computation rate that can be achieved using this Reduce-and-Broadcast scheme can be expressed as the optimal solution to a linear program and this result is stated in the following theorem.
\begin{theorem}\label{thm:lp}
For network $\mathcal{N}(\vec{\beta})$, the All-Reduce computation rate satisfies
\begin{eqnarray}
&& R^* \geq \underline{R} ~\triangleq~ \max \lambda_1 + \cdots + \lambda_Z \notag \\
&& ~~~~~~~~~~~~~~ s.t.~~ \lambda_1 \vec{\beta}_{T_{MB}}^1 + \cdots + \lambda_Z \vec{\beta}_{T_{MB}}^Z \leq \vec{\beta} \label{eq:lp} \\
&& ~~~~~~~~~~~~~~~~~~~~~ \lambda_1, \cdots, \lambda_Z \in \mathbb{R}_+. \notag
\end{eqnarray}
\end{theorem}

{\it Proof:} The proof is almost immediate. For each rooted MAC-BC network, we may achieve computation rate of $1$, then by time (bandwidth) sharing over all $Z$ rooted MAC-BC networks with weights $\lambda_1, \cdots, \lambda_Z$, the achieved rate becomes $\lambda_1 + \cdots + \lambda_Z$ subject to the constraint that each weight $\lambda_z$ is non-negative (the solution to the linear program (\ref{eq:lp}) must be rational as the bandwidth vectors have integer elements) and the bandwidth of each link can support the communication flow.

\hfill\qed

Solving the optimal solution of the linear program (\ref{eq:lp}) is computationally heavy (infeasible in general) because $Z = K^{2K-3}$ is prohibitively large and finding all $\vec{\beta}^z_{T_{MB}}$ explicitly is highly non-trivial. Fortunately, we have an explicit simple upper bound to $\underline{R}$ (that only depends on $\vec{\beta}$) so that whenever the upper bound is achieved (usually by using far less than $Z$ rooted MAC-BC networks), we know the optimal solution to linear program (\ref{eq:lp}) (this upper bound will be used frequently in the following sections). This result in stated in the following corollary.


\begin{corollary}\label{cor}
For any $\vec{\beta} = (\beta_{12}; \cdots; \beta_{(K-1) K})$, 
\begin{eqnarray}
\underline{R} \leq \frac{ \beta_{12} + \cdots + \beta_{(K-1)K} }{2(K-1)}. \label{eq:lpbound}
\end{eqnarray}
\end{corollary}

{\it Proof:} The key observation is that the sum of all elements in $\vec{\beta}^z_{T_{MB}}$ is the constant $2(K-1)$ for any $z \in [Z]$ because each rooted MAC and BC tree network has exactly $K-1$ unit bandwidth links (edges). Then taking the sum of all rows of (\ref{eq:lp}) gives $(\lambda_1 + \cdots + \lambda_Z) \times 2(K-1) \leq \beta_{12} + \cdots + \beta_{(K-1) K}$ and we have the desired bound (\ref{eq:lpbound}).

\hfill\qed


\subsection{A Class of Networks where $R^* = \overline{R} = \underline{R}$}\label{sec:optimal}

The starting point of this section is the observation that for any rooted MAC-BC network $\mathcal{N}_{T_{MB}}(\vec{\beta}_{T_{MB}})$ (where the root is not shown as it is not important here), the maximum computation rate is $R^* = \underline{R} = \overline{R} = 1$. $\underline{R} = 1$ as Reduce-and-Broadcast can achieve rate $1$. $\overline{R} = 1$ can be seen as follows. $\mathcal{N}_{T_{MB}}(\vec{\beta}_{T_{MB}})$ contains exactly $2(K-1)$ edges (unit bandwidth links where $\beta_{ij} = 1$) so that on average the in-degree (out-degree) of each node is $\frac{2(K-1)}{K} = 2 - 1/K$. Then by the pigeonhole principle, there must exist a node $k \in [K]$ whose in-degree (out-degree) is $1$, i.e., there exists $k^* \neq k, k^* \in [K]$ so that $\beta_{k^*k} = 1$ while $\beta_{i k} = 0$ for all $i \neq k^*$ ($\beta_{kk^*} = 1$ while $\beta_{k j} = 0$ for all $j \neq k^*$). Then setting $\mathcal{K} = [K] \setminus \{i\}$ ($\mathcal{K} = \{i\}$) in (\ref{eq:cut}) gives $\overline{R} =  1$. As the edge $k^* \rightarrow k$ ($k \rightarrow k^*$) forms a cut (and is critical), we call it a {\em cut-edge} and include it in the notation of a rooted MAC-BC network, i.e., $\mathcal{N}_{T_{MB}}(\vec{\beta}_{T_{MB}}, i^* \rightarrow j^*)$ where $i^* \rightarrow j^*$ is a cut-edge, i.e., either it is the only edge entering $j^*$ or the only edge leaving $i^*$.

Next, building upon the above optimality result of a rooted MAC-BC network $\mathcal{N}_{T_{MB}}(\vec{\beta}_{T_{MB}}, i^* \rightarrow j^*)$, we wish to generalize it to a larger class of networks where we keep the bandwidth of the cut-edge and may increase the bandwidth for other non-zero links so that $\overline{R}$ remains $1$ and $\underline{R} \geq 1$ as increasing bandwidth cannot hurt. This leads us to the following definition of $1$-MAC-BC networks.



\begin{definition}[1-MAC-BC Network]
Given a rooted MAC-BC network $\mathcal{N}_{T_{MB}}(\vec{\beta}_{T_{MB}}, i^* \rightarrow j^*)$ with cut-edge $i^* \rightarrow j^*$, we call a network $\mathcal{N}_{MB}(\vec{\beta}_{MB}, i^* \rightarrow j^*)$ where 
\begin{eqnarray}
&& \beta_{ij} = 0 ~\mbox{in}~ \vec{\beta}_{MB} ~\mbox{if}~ \beta_{ij} = 0 ~\mbox{in}~ \vec{\beta}_{T_{MB}} \\
&& \beta_{i^* j^*} = 1 ~\mbox{in}~ \vec{\beta}_{MB}\\
&& \vec{\beta}_{MB} \geq \vec{\beta}_{T_{MB}}
\end{eqnarray}
a $1$-MAC-BC network.
\end{definition}

An example of $1$-MAC-BC network is shown in Fig.~\ref{fig:1macbc}.




\begin{figure}[h]
\centering
\includegraphics[scale=0.7]{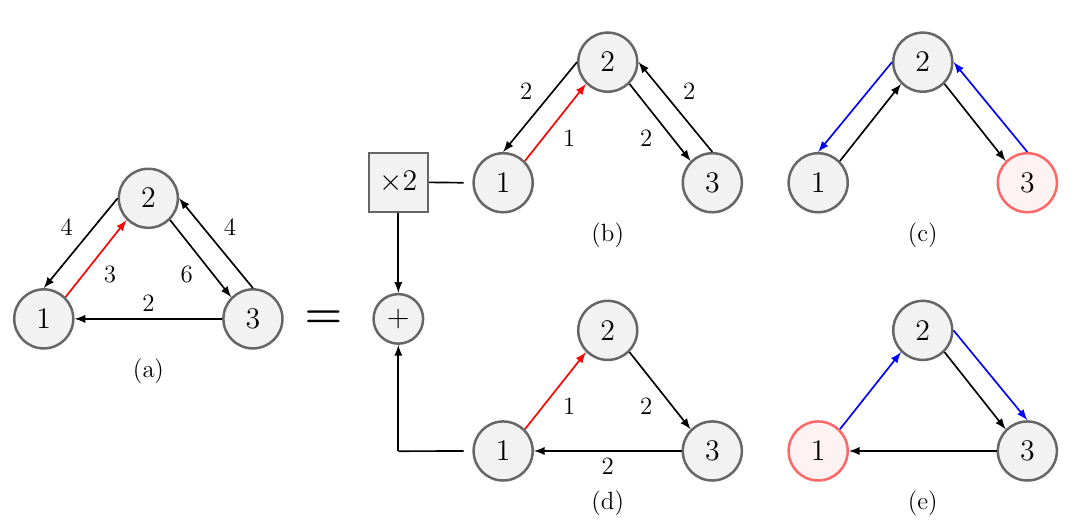}
\caption{ (b) is a $1$-MAC-BC network $\mathcal{N}_{MB}^1(\vec{\beta}_{MB}^1, 1 \rightarrow 2)$ with its associated rooted MAC-BC network in (d). (c) is a $1$-MAC-BC network $\mathcal{N}_{MB}^2(\vec{\beta}_{MB}^2, 1 \rightarrow 2)$ with its associated rooted MAC-BC network in (e). The cut-edge $1 \rightarrow 2$ is colored in red. (a) contains a network $\mathcal{N}(\vec{\beta})$ that is a linear combination of $\mathcal{N}_{MB}^1(\vec{\beta}_{MB}^1, 1 \rightarrow 2)$ and $\mathcal{N}_{MB}^2(\vec{\beta}_{MB}^2, 1 \rightarrow 2)$ considered in Theorem \ref{thm:1macbc} where $\vec{\beta} = 2 \vec{\beta}_{MB}^1 + \vec{\beta}_{MB}^2$.
}
\label{fig:1macbc}
\end{figure}




We may further generalize the optimality of $1$-MAC-BC network to their linear combinations. This result is stated in the following theorem.





\begin{theorem}\label{thm:1macbc}
Consider $U$ $1$-MAC-BC networks with $K$ nodes and the same cut-edge, $\mathcal{N}_{MB}^1(\vec{\beta}^1_{MB}, i^* \rightarrow j^*)$, $\cdots$, $\mathcal{N}_{MB}^U (\vec{\beta}^U_{MB}, i^* \rightarrow j^*)$, then the maximum All-Reduce computation rate of network $\mathcal{N}(\vec{\beta})$ where $\vec{\beta} = \lambda_1 \vec{\beta}^1_{MB} + \cdots + \lambda_U \vec{\beta}^U_{MB}$ is $\lambda_1 + \cdots + \lambda_U$.
\end{theorem}

{\it Proof:}
Almost immediate from above discussions. 
Consider network $\mathcal{N}(\vec{\beta})$ where $\vec{\beta} = \lambda_1 \vec{\beta}^1_{MB} + \cdots + \lambda_U \vec{\beta}^U_{MB}$. $\underline{R} \geq \lambda_1 + \cdots + \lambda_U$ as computation rate $1$ can be achieved for each $1$-MAC-BC network and then we may apply bandwidth sharing over the $U$ $1$-MAC-BC networks (refer to Theorem \ref{thm:lp}). $\overline{R} \leq \lambda_1 + \cdots + \lambda_U$ as $i^* \rightarrow j^*$ continues to be a cut-edge for $\mathcal{N}(\vec{\beta})$ so that we may apply Theorem \ref{thm:cut}, i.e., (\ref{eq:cut}) to obtain $\overline{R} \leq \beta_{i^*j^*} = \lambda_1 + \cdots + \lambda_U$.
\hfill\qed

Theorem \ref{thm:1macbc} covers all networks $\mathcal{N}(\vec{\beta})$ whose topology is a bi-directed tree as a special case. See Fig.~\ref{fig:tree} for an example.

\begin{figure}[hbt!]
\centering
\includegraphics[scale=0.75]{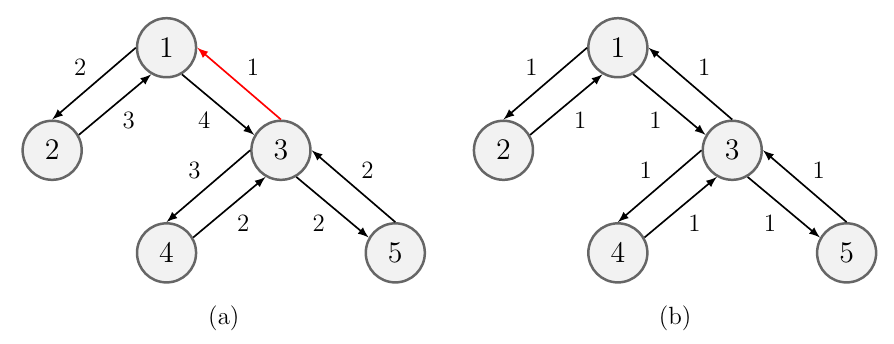}
\caption{(a). A network $\mathcal{N}(\vec{\beta})$ whose topology is a bi-directed tree and (b). its associated rooted MAC-BC network. The capacity of a bi-directed tree is equal to the minimum bandwidth of all links.}
\label{fig:tree}
\end{figure}

\subsection{Specific Topologies with Explicit $\underline{R}, \overline{R}$}\label{sec:ach}

In this section, we evaluate the upper bound $\overline{R}$ and the lower bound $\underline{R}$ explicitly for a few classes of networks $\mathcal{N}(\vec{\beta})$ and for every case, $\overline{R} \leq 2 \underline{R}$ so that the upper and lower bound of $R^*$ is within a multiplicative gap of $2$.

\subsubsection{Uniform Complete Networks}

A network $\mathcal{N}_C(\vec{\beta}_C)$ is called uniform complete if
\begin{eqnarray}
\beta_{ij} = 1, \forall i, j \in [K], i \neq j.
\end{eqnarray}

\begin{theorem}
For a uniform complete network $\mathcal{N}_C(\vec{\beta}_C)$, 
\begin{eqnarray}
K/2 \leq R^* \leq K-1.
\end{eqnarray}
\end{theorem}

{\it Proof:} $\overline{R} \leq K-1$ as we may set $\mathcal{K} = \{1\}$ in (\ref{eq:cut}) and apply Theorem \ref{thm:cut}, i.e., $R^* \leq \sum_{j = 2}^K \beta_{1j} = K-1$. Further, it can be proved that $\overline{R} = K-1$, i.e., $K-1$ is the tightest upper bound in (\ref{eq:cut}).

To prove $\underline{R} = K/2$, we may use the $K$ rooted MAC-BC networks shown in Fig.~\ref{fig:complete}, each with weight $\lambda_k = 1/2, k \in [K]$. Then the bandwidth constraint (\ref{eq:lp}) is satisfied and $\sum_{k\in[K]} \lambda_k = K/2$, which is equal to $\frac{\sum_{i,j} \beta_{ij}}{2(K-1)} = \frac{K(K-1)}{2(K-1)}$ so that by Corollary \ref{cor}, we have $\underline{R} = K/2$.



\begin{figure}[h]
\centering
\includegraphics[scale=0.7]{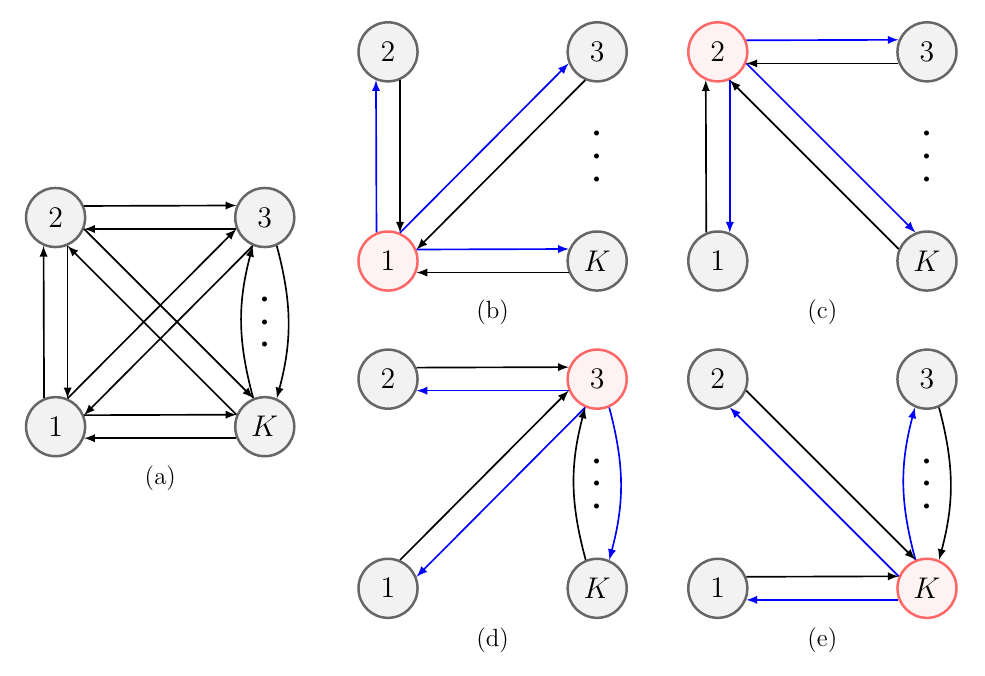}
\caption{(a). Uniform complete network and (b)-(e). $K$ rooted MAC-BC networks used.}
\label{fig:complete}
\end{figure}

\hfill\qed





\subsubsection{Cycles and Rings}

A network $\mathcal{N}_{cyc}(\vec{\beta}_{cyc})$ is called uniform cyclic if
\begin{eqnarray}
&& \beta_{12} = \beta_{23} = \cdots = \beta_{(K-1)K} = \beta_{K 1} = 1 \\
&& \mbox{other $\beta_{ij} = 0$}.
\end{eqnarray}
A network 
$\mathcal{N}_{r}(\vec{\beta}_{r})$ is called uniform ring if
\begin{eqnarray}
&& \beta_{12} = \beta_{23} = \cdots = \beta_{(K-1)K} = \beta_{K 1} = 1 \\
&& \beta_{21} = \beta_{32} = \cdots = \beta_{K(K-1)} = \beta_{1 K} = 1 \\
&& \mbox{other $\beta_{ij} = 0$}.
\end{eqnarray}

\begin{theorem}
For a uniform cyclic network $\mathcal{N}_{cyc}(\vec{\beta}_{cyc})$, 
\begin{eqnarray}
\frac{K}{2(K-1)} \leq R^* \leq 1.
\end{eqnarray}
For a uniform ring network $\mathcal{N}_{r}(\vec{\beta}_{r})$, 
\begin{eqnarray}
\frac{K}{K-1} \leq R^* \leq 2.
\end{eqnarray}
\end{theorem}

{\it Remark: The achieved rate for the uniform ring network is the same as that of the famous Ring-All-Reduce scheme \cite{patarasuk2009bandwidth, patarasuk2007bandwidth}, which applies first Reduce-Scatter and then All-Gather.}

{\it Proof:} We prove the result for the uniform cycle network and that for the ring network follows as it contains two uniform cycles (one clockwise and one counter-clockwise). 

For $\mathcal{N}_{cyc}(\vec{\beta}_{cyc})$, $\overline{R} = 1$ as we may set $\mathcal{K} = \{1\}$ in (\ref{eq:cut}). 

$\underline{R} = \frac{K}{2(K-1)}$ as we may use the $K$ rooted MAC-BC networks in Fig.~\ref{fig:cycle}, each with weight $\lambda_k = \frac{1}{2(K-1)}, k \in [K]$ and apply Corollary \ref{cor}.

\begin{figure}[hbt!]
\centering
\includegraphics[scale=0.7]{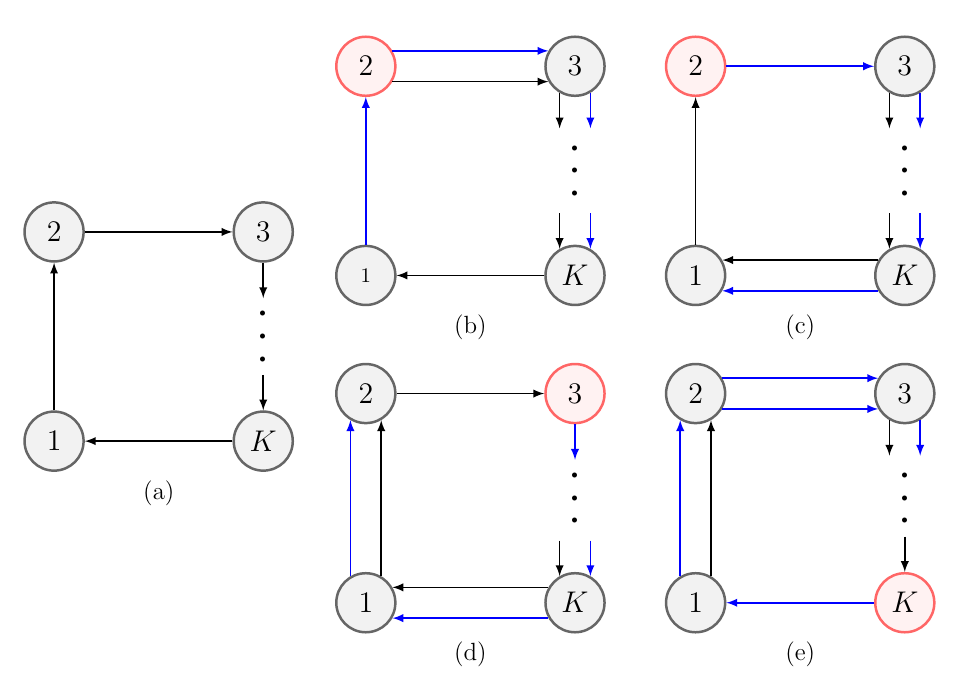}
\caption{(a). Uniform cyclic network and (b)-(e). $K$ rooted MAC-BC networks used.}
\label{fig:cycle}
\end{figure}

\hfill\qed


We consider a $K=3$ node non-uniform cyclic network $\mathcal{N}_{3cyc}(\vec{\beta}_{3cyc})$ where 
\begin{eqnarray}
\beta_{12} = a, \beta_{23} = b, \beta_{31} = c, \beta_{21} = \beta_{32} = \beta_{13} = 0.
\end{eqnarray}

\begin{theorem}
For a $3$-node cyclic network $\mathcal{N}_{3cyc}(a;b;c)$ (the $\vec{\beta}$ notation is simplified where zero elements are not shown),
\begin{eqnarray}
&& \mbox{if}~\min(a,b,c) \leq \frac{a+b+c}{4}, ~\mbox{then}~R^* = \min(a,b,c), \\
&& \mbox{else}~\min(a,b,c) > \frac{a+b+c}{4}, ~\mbox{then}~\frac{a+b+c}{4} \leq R^* \leq \min(a,b,c).
\end{eqnarray}
\end{theorem}




{\it Remark: When $\min(a,b,c) \leq \frac{a+b+c}{4}$, the optimal computation rate is characterized; when $\min(a,b,c) > \frac{a+b+c}{4}$, the optimal computation rate is characterized to within a multiplicative gap of $4/3$ as $\min(a,b,c) \leq \frac{3}{4} \times \frac{a+b+c}{4}$.}

{\it Proof:} From Theorem \ref{thm:cut}, we have $\overline{R} = \min(a,b,c)$ and the upper bound proof is complete. Next we proceed to the lower bound proof. We use the following three MAC-BC networks, $\mathcal{N}_1(2;1;1)$, $\mathcal{N}_2(1;2;1)$, and $\mathcal{N}_3(1;1;2)$ with weights $\lambda_1, \lambda_2, \lambda_3$, respectively.

First, consider $\min(a,b,c) \leq \frac{a+b+c}{4}$. By symmetry, we may assume $a \leq b \leq c$ without loss of generality. 
We show that $\underline{R} \geq a$. Set
\begin{eqnarray}
&& \lambda_1 = 0, \lambda_2 = 0, \lambda_3 = a ~\mbox{if}~ c \geq 2a \\
&& \lambda_1 = 0, \lambda_2 = 2a - c, \lambda_3 = c - a ~\mbox{else}~ c < 2a
\end{eqnarray}
where $\lambda_1, \lambda_2, \lambda_3$ are non-negative and the bandwidth constraint (\ref{eq:lp}) is satisfied because (the $c \geq 2a$ case is obvious and we only consider the $c < 2a$ case)
\begin{eqnarray}
&& \beta_{12}: \lambda_2 \times 1 + \lambda_3 \times 1 = a \\
&& \beta_{23}: \lambda_2 \times 2 + \lambda_3 \times 1 = 3a - c \leq b ~\mbox{as}~ a \leq \frac{a+b+c}{4}\\
&& \beta_{31}: \lambda_2 \times 1 + \lambda_3 \times 2 = c
\end{eqnarray}
So $\underline{R} \geq \lambda_1 + \lambda_2 + \lambda_3 = a$.


Second, consider $\min(a,b,c) > \frac{a+b+c}{4}$ and we show that $\underline{R} = \frac{a+b+c}{4}$. Set
\begin{eqnarray}
\lambda_1 = \frac{3a-b-c}{4}, \lambda_2 = \frac{3b-a-c}{4}, \lambda_3 = \frac{3c-a-b}{4}
\end{eqnarray}
where $\lambda_1 > 0, \lambda_2 > 0, \lambda_3 > 0$ due to $\min(a,b,c) > \frac{a+b+c}{4}$ and the bandwidth constraint (\ref{eq:lp}) is satisfied because 
\begin{eqnarray}
&& \beta_{12}: \lambda_1 \times 2 + \lambda_2 \times 1 + \lambda_3 \times 1 = a \\
&& \beta_{23}: \lambda_1 \times 1 + \lambda_2 \times 2 + \lambda_3 \times 1 = b \\
&& \beta_{31}:  \lambda_1 \times 1 + \lambda_2 \times 1 + \lambda_3 \times 2 = c
\end{eqnarray}
So $\underline{R} \geq \lambda_1 + \lambda_2 + \lambda_3 = \frac{a+b+c}{3}$ ($\geq$ can be replaced by $=$ due to Corollary \ref{cor}).
%
\hfill\qed

\subsubsection{Uniform Hypercubes}


A network $\mathcal{N}_{H}(\vec{\beta}_{H})$ is called uniform hypercube if
\begin{eqnarray}
&& K = 2^U, U \in \mathbb{N}~\mbox{and label the nodes as $1,\cdots, K-1, 0$ (i.e., $K$ is replaced by $0$)} \notag \\
&& \beta_{ij} = 1 ~\mbox{if the binary representations of $i, j$ differ in one bit}, ~i,j \in \{0,1,\cdots,K-1\} \label{eq:cube}\\
&& \mbox{other $\beta_{ij} = 0$}. \notag
\end{eqnarray}

\begin{theorem}
For a uniform hypercube network $\mathcal{N}_{H}(\vec{\beta}_{H})$ with $K = 2^U$ nodes,
\begin{eqnarray}
\frac{2^{U-1}}{2^U-1} U   \leq R^* \leq U.
\end{eqnarray}
\end{theorem}

{\it Remark:
The above (achieved) rate lower bound is higher than that of the scheme in Section 13.2.1 of \cite{sanders2019sequential} (which is better in other metrics, such as latency).
}

{\it Proof:} $\overline{R} = U$ follows from setting $\mathcal{K} = \{0\}$ in (\ref{eq:cut}) (note that each node is connected to $U$ other nodes in a hypercube) and Theorem \ref{thm:cut}.



Next, we show that $\underline{R} = \frac{2^{U-1}}{2^U-1} U$.
We will use $2^U \times U!$ rooted MAC-BC networks $\mathcal{N}_{T_{MB}}^\pi (\vec{\beta}^\pi_{T_{MB}}, r)$, where $r$ takes all values from the node set $\{0,1,\cdots, K-1\}$ and $\pi = (\pi_1, \cdots, \pi_U)$ takes all $U!$ permutations of $\{1,\cdots, U\}$. For a fixed $r$ and a fixed permutation $\pi$, the rooted MAC-BC network is specified as follows. Denote the binary representation of $r$ as $r = (r_1, \cdots, r_U)$ where $r_u \in \{0,1\}, u \in [U]$. Consider the $K$ nodes except the root in $U$ steps and denote the node set considered in the $u$-th step by $\mathcal{K}_u$ where $|\mathcal{K}_u| = 2^{u-1}$. $\mathcal{K}_u$ are set as follows.
\begin{eqnarray}
&& \mathcal{N}_{T_{MB}}^\pi (\vec{\beta}^\pi_{T_{MB}}, r): \pi = (\pi_1, \cdots, \pi_U), r = (r_1, \cdots, r_U)\\
&& 
\mathcal{K}_0: \{r\} \\
&& 
\mathcal{K}_1: \{k\}, ~\mbox{where}~ k = (r_1, \cdots, 1 - r_{\pi_1}, \cdots, r_U), ~\mbox{i.e., flip the $\pi_1$-th bit of $r$} \label{c1} \\
&& 
\mathcal{K}_2: (k_1, \cdots, 1-k_{\pi_2}, \cdots, k_U) \in \mathcal{K}_2 ~\mbox{for each}~k = (k_1, \cdots, k_U) \in \mathcal{K}_0 \cup \mathcal{K}_1 \notag\\
&&~~~~~~~~~\mbox{i.e., flip the $\pi_2$-th bit of all nodes appeared so far}  \label{c2}\\
&& ~~\vdots \label{c3}\\
&& 
\mathcal{K}_{U}: (k_1, \cdots, 1-k_{\pi_U}, \cdots, k_U) \in \mathcal{K}_{U} ~\mbox{for each}~k = (k_1, \cdots, k_U) \in \mathcal{K}_0 \cup \cdots \cup \mathcal{K}_{U-1} \label{c4}
\end{eqnarray}

Then the unit bandwidth links of $\mathcal{N}_{T_{MB}}^\pi (\vec{\beta}^\pi_{T_{MB}}, r)$ are all edges that connect each node in $\mathcal{K}_u$ to all nodes in $\mathcal{K}_0 \cup \cdots \cup \mathcal{K}_{u-1}$ whenever allowed in hypercube.
\begin{eqnarray}
&& \mathcal{N}_{T_{MB}}^\pi (\vec{\beta}^\pi_{T_{MB}}, r): \beta_{ij} = \beta_{ji} = 1 ~\mbox{if}~ i \in \mathcal{K}_u, j \in  \mathcal{K}_{u^*}, u, u^* \in \{1,\cdots, U\}, u^* < u, ~\mbox{and $i,j$ differ in one bit} \notag \\
&& ~~~~~~~~~~~~~~~~~~~~~~~~~~~~~~ \mbox{other}~\beta_{ij} = 0. \label{c5}
\end{eqnarray}

An example of $\mathcal{N}_{T_{MB}}^\pi (\vec{\beta}^\pi_{T_{MB}}, r)$ is shown in Fig.~\ref{fig:cube}.

\begin{figure}[hbt!]
\centering
\includegraphics[scale=0.95]{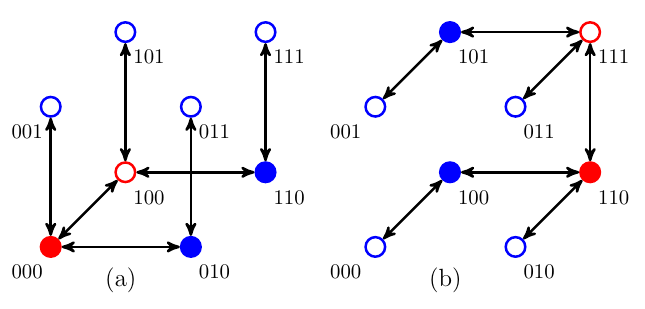}
\caption{The rooted MAC-BC network $\mathcal{N}_{T_{MB}}^\pi (\vec{\beta}^\pi_{T_{MB}}, r)$ when (a). $r = (0,0,0)$, $\pi = (1,2,3)$ and (b). $r=(1,1,0)$, $\pi = (3,2,1)$. The red filled circle denotes the node in $\mathcal{K}_0$ (the root), the red hollow circle denotes the node in $\mathcal{K}_1$, the blue filled circles denote the nodes in $\mathcal{K}_2$, and the blue hollow circles denote the nodes in $\mathcal{K}_3$.}
\label{fig:cube}
\end{figure}


Each one of the $2^U \times U! \triangleq Z$ rooted MAC-BC networks will be associated with weight $\lambda_z = \frac{1}{2(2^U-1)(U-1)!}, z \in [Z]$
so that $\overline{R} = \sum_{z\in[Z]} \lambda_z = \frac{2^{U-1} \times U}{2^{U}-1} = \frac{K \times U}{2(K-1)} = \frac{\sum \beta_{ij}}{2(K-1)}$. We are left to prove the bandwidth constraint (\ref{eq:lp}) is satisfied. To this end, we need to show that each edge $i \rightarrow j$ of the hypercube $\mathcal{N}_H$ appears in the $Z$ rooted MAC-BC networks $\mathcal{N}_{T_{MB}}^\pi (\vec{\beta}^\pi_{T_{MB}}, r)$ for $2(2^U-1)(U-1)! = 1/\lambda_z$ times in total.

Before considering the general case, we first give an example of $2^3$-node hypercube and edge $(0,0,0) \rightarrow (1,1,1)$ to illustrate the idea in a simpler setting.





\begin{example}
Consider the rooted MAC-BC networks $\mathcal{N}_{T_{MB}}^\pi (\vec{\beta}^\pi_{T_{MB}}, r), \forall r \in \{(0,0,0), \cdots, (1,1,1)\}, \forall \pi = (\pi_1, \pi_2, \pi_3)$ that is a permutation of $\{1,2,3\}$ and we check which $\mathcal{N}_{T_{MB}}^\pi$ contain edge $(0,0,0) \rightarrow (0,0,1)$ as a unit bandwidth link.

Note that in edge $(0,0,0) \rightarrow (1,1,1)$, the two nodes only differ in the $3$-rd bit representation. We classify all $\mathcal{N}_{T_{MB}}^\pi$ by which $\pi_u, u \in \{1,2,3\}$ is equal to $3$ and have the following $3$ cases.

\begin{enumerate}
\item $\pi_1 = 3$: According to which links have unit bandwidth in $\mathcal{N}_{T_{MB}}^\pi (\vec{\beta}^\pi_{T_{MB}}, r)$ (refer to (\ref{c1}) - (\ref{c5})), in order for edge $(0,0,0) \rightarrow (0,0,1)$ to appear in $\mathcal{N}_{T_{MB}}^\pi (\vec{\beta}^\pi_{T_{MB}}, r)$, $(0,0,0)$ and $(0,0,1)$ must be in $\mathcal{K}_0$ and $\mathcal{K}_1$ (as only the $3$-rd bit changes and $\pi_1 = 3$). Therefore when $\pi_1 = 3$, edge $(0,0,0) \rightarrow (0,0,1)$ appears only in $\mathcal{N}_{T_{MB}}^\pi (\vec{\beta}^\pi_{T_{MB}}, r)$ where $\pi_1 = 3$ and $r = \{0,0,0\}$ or $r = \{0,0,1\}$, i.e., $2! \times 2 = 4$ times where $2!$ corresponds to the number of choices of $\pi_2, \pi_3$ and $2$ corresponds to the number of choices of $r$. See Fig.~\ref{fig:edge}.(a) for an illustration. 



\item $\pi_2 = 3$: Here for edge $(0,0,0) \rightarrow (0,0,1)$ to appear in $\mathcal{N}_{T_{MB}}^\pi (\vec{\beta}^\pi_{T_{MB}}, r)$, according to (\ref{c1}) - (\ref{c5}), one of $(0,0,0)$ and $(0,0,1)$ must be from $\mathcal{K}_0 \cup \mathcal{K}_1$ and the other one must be from $\mathcal{K}_2$ (edges where the $3$-rd bit changes must be connected to $\mathcal{K}_2$ nodes). First, suppose $(0,0,0) \in \mathcal{K}_2$. $\pi_1, \pi_3$ can take $2!$ choices. For either case, say $\pi_1 = 1, \pi_3 = 2$, the root has $2$ choices, i.e., $r = (1,0,1)$ or $r = (0,0,1)$ where the $1$-st bit can be arbitrary. So overall the edge appears $2! \times 2 = 4$ times. 
Second, suppose $(0,0,1) \in \mathcal{K}_2$. This case can be considered similarly and edge $(0,0,0) \rightarrow (0,0,1)$ also appears $4$ times. Combining both we have $(0,0,0) \rightarrow (0,0,1)$ appears $8$ times. See Fig.~\ref{fig:edge}.(b) for an illustration. 


\item $\pi_3 = 3$: Here for edge $(0,0,0) \rightarrow (0,0,1)$ to appear in $\mathcal{N}_{T_{MB}}^\pi (\vec{\beta}^\pi_{T_{MB}}, r)$, one of $(0,0,0)$ and $(0,0,1)$ must be in $\mathcal{K}_3$ (say $(0,0,0)$) and the other one must be from $\mathcal{K}_0 \cup \mathcal{K}_1 \cup \mathcal{K}_2$. $\pi_1, \pi_2$ can take $2!$ choices and for every choice, the root $r$ can take any value from the set $\{(0,0,1), (0,1,1), (1,0,1), (1,1,1)\}$ where the first $2$ bits can be arbitrary. So overall edge $(0,0,0) \rightarrow (0,0,1)$ appears $2 \times 2! \times 2^2 = 16$ times. See Fig.~\ref{fig:edge}.(c) for an illustration. 

\end{enumerate}

To sum up, in total edge $(0,0,0) \rightarrow (0,0,1)$ appears $4+8+16 = 28$ times, as desired (note that $28 = 2 \times (2^3-1) \times 2!$). 


\begin{figure}[hbt!]
\centering
\includegraphics[scale=0.7]{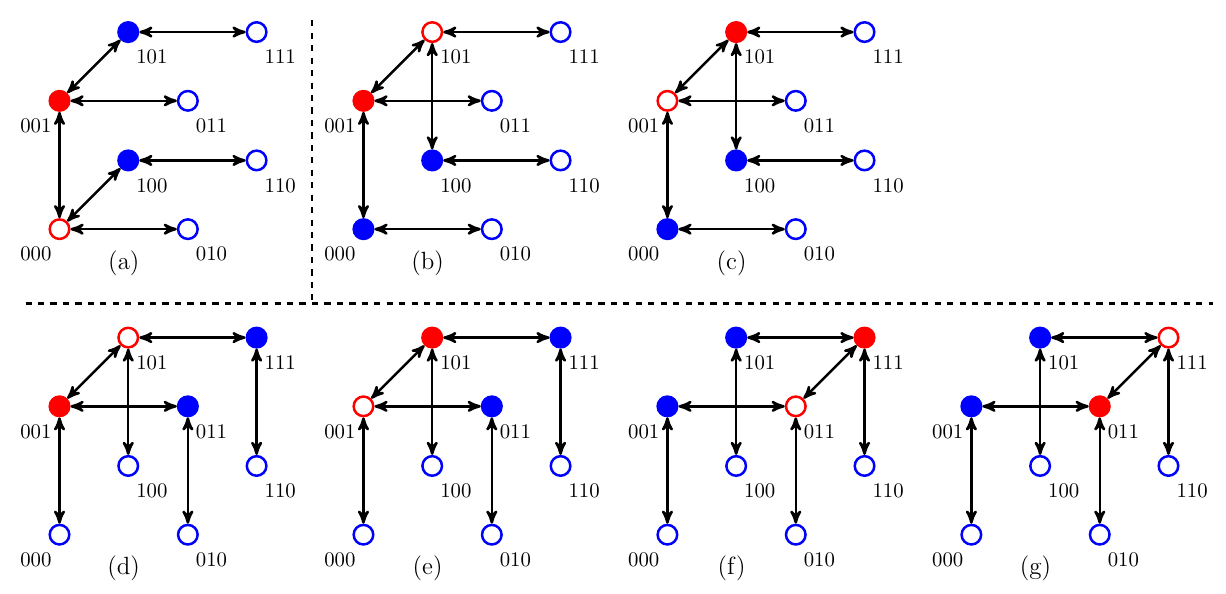}
\caption{$\mathcal{N}_{T_{MB}}^\pi (\vec{\beta}^\pi_{T_{MB}}, r)$ where edge $(0,0,0) \rightarrow (0,0,1)$ appears. (a). An example of $\pi_1 = 3$: $\pi=(3,1,2)$, $\mathcal{K}_1 = \{(0,0,0)\}$, and $\mathcal{K}_0 = \{(0,0,1)\}$. (b) - (c). Examples of $\pi_2 = 3$: $\pi=(1,3,2)$, $(0,0,0) \in \mathcal{K}_2$, then $\mathcal{K}_1 = \{(r_1,0,1)\}$, and $\mathcal{K}_0 = \{(1-r_1,0,1)\}, r_1 \in \{0,1\}$. (d) - (g). Examples of $\pi_3 = 3$: $\pi=(1,2,3)$, $(0,0,0) \in \mathcal{K}_3$, then $(r_1,r_2,1) \in \mathcal{K}_2$, $\mathcal{K}_1 = \{(r_1,1-r_2,1)\}$, and $\mathcal{K}_0 = \{(1-r_1,1-r_2,1)\}, r_1, r_2 \in \{0,1\}$.}
\label{fig:edge}
\end{figure}
\end{example}

We now proceed to the general case, i.e., any edge (unit bandwidth link) $(k_1, \cdots, k_u, \cdots, k_U) \rightarrow (k_1, \cdots, 1-k_u, \cdots, k_U)$ of $\mathcal{N}_H$ appears for a total number of $2(2^{U}-1)(U-1)!$ times in all $Z$ rooted MAC-BC networks $\mathcal{N}_{T_{MB}}^\pi (\vec{\beta}^\pi_{T_{MB}}, r)$. 

Similar to the example above, in edge $(k_1, \cdots, k_u, \cdots, k_U) \rightarrow (k_1, \cdots, 1-k_u, \cdots, k_U)$, the $u$-th bit of the two nodes differs and consider the cases where $\pi_1 = u$, or $\pi_2 = u$, $\cdots$, or $\pi_U = u$.

When $\pi_d = u, d \in [U]$, one of $(k_1, \cdots, k_u, \cdots, k_U)$ and $(k_1, \cdots, 1-k_u, \cdots, k_U)$ is in $\mathcal{K}_d$ and the other one is in $\mathcal{K}_{d^*}$ where $d^* < d$. For either case (say $(k_1, \cdots, k_u, \cdots, k_U) \in \mathcal{K}_d$), $(\pi_1, \cdots, \pi_{d-1}, \pi_{d+1}, \cdots, \pi_U)$ may take $(U-1)!$ choices and for any choice, $r$ may take $2^{d-1}$ choices, i.e., $r = (r_1, \cdots, r_{d-1}, k_{d}, \cdots, k_U)$ where $r_1, \cdots, r_{d-1}$ can be arbitrary. Counting everything, we have
\begin{eqnarray}
\sum_{d=1}^U 2 \times (U-1)! \times 2^{d-1} = 2 \times (U-1)! \times (2^0 + 2^1 + \cdots + 2^{U-1}) = 2 (2^U -1) (U-1)!
\end{eqnarray}
and the proof is complete.
\hfill\qed

\section{Discussion}\label{sec:dis}

We have studied the All-Reduce problem from a computation rate perspective. Two simple general bounds are provided - the cut-set upper bound and the Reduce-Broadcast-LP lower bound. We cannot find a single network instance where we may improve either the upper or lower bound. 

For the upper bound side, note that the network considered in this work contains cycles while most prior work in network coding consider acyclic networks so that we do not have many tools for converse and cyclic networks are tricky to deal with \cite{harvey2007chicken, erez2010efficient}. In addition to cycles in the network, note that All-Reduce is a computation problem; further, a compound computation problem where everyone (so more than $1$ node/destination) demands the sum. This compound computation aspect is also under-explored in network information theory \cite{NIT} and might add some complexity to the problem (so calls for more attention).

For the lower bound side, the Reduce-Broadcast scheme is not the only possible scheme as we may alternatively first Reduce-Scatter and then All-Gather as in the widely applied Ring-All-Reduce scheme \cite{patarasuk2009bandwidth, patarasuk2007bandwidth}. 
Separate coding is used in this work and throughout the literature and it is interesting to see if joint coding over different sum instances strictly helps (note that Ring-All-Reduce, whose achieved rate is the same as that of our Reduce-Broadcast-LP lower bound has been proved optimal for a class of separate coding schemes \cite{patarasuk2009bandwidth, patarasuk2007bandwidth}, referred to as bandwidth optimal in computer science literature, but this is weaker than information theoretic optimality as considered in this work).

Probably the smallest open problem is the $3$-node uniform complete network (coinciding with the $3$-node uniform ring network) where $3/2 \leq R^* \leq 2$. It is not clear if either bound may be tight. 
In addition to tighten the two bounds, another interesting open question is to see if the multiplicative gap of $2$ of the two bounds obtained for all the networks considered in this work can be generalized to all networks, i.e., if $\overline{R} \leq 2\underline{R}$ for all $\vec{\beta}$ and if not, finding tighter bounds with possibly even smaller gaps in general. Finally, All-Reduce becomes intimately related to secure aggregation \cite{zhao2022information, sun2023secure} after including security constraints and it is of interest to characterize the secure computation rate of All-Reduce.  

\bibliography{Thesis}
\end{document}